\documentclass[letterpaper, 10 pt, conference]{ieeeconf}

\IEEEoverridecommandlockouts

\usepackage{mathptmx} 
 \usepackage{times} 
\usepackage[utf8]{inputenc}
\usepackage{amssymb}
\usepackage{amsmath}
\usepackage{mathrsfs}
\usepackage{graphicx}
\usepackage{hyperref}
\usepackage{bbm}
\usepackage{algorithm}
\usepackage[noend]{algpseudocode}
\usepackage{cite}
\usepackage{amsmath,amssymb,amsfonts,amsthm}
\usepackage{textcomp}
\usepackage{epstopdf}
\usepackage{tikz}
\usetikzlibrary{calc}
\usetikzlibrary{matrix, positioning, fit, calc}
\usepackage{tkz-euclide}
\usepackage{graphics} 
\usepackage{epsfig} 
\usepackage{cite}  
\usepackage{xcolor}
\usepackage{bbm}
\usepackage{enumerate}
\usepackage{algorithm}
\usepackage{amstext}
\usepackage{subfig}
\usepackage{multicol}

\theoremstyle{plain}

\newtheorem{assumption}{Assumption}

\newtheorem{proposition}{Proposition}

\newtheorem{lemma}{Lemma}

\newtheorem{theorem}{Theorem}

\newcommand{\G}{\mathcal{G}}
\newcommand{\V}{\mathcal{V}}

\newcommand{\Ls}{\mathcal{L}}
\newcommand{\E}{\mathcal{E}}
\newcommand{\K}{\mathcal{K}}

\newcommand{\M}{\mathcal{M}}

\newcommand{\N}{\mathcal{N}}

\renewcommand{\L}{\mathbf{L}}

\usepackage{accents}

\newcommand*{\dif}{\mathop{}\!\mathrm{d}}

\newcommand{\real}{\mathbb{R}}
\newcommand{\complex}{\mathbb{C}}

\newcommand{\naturals}{\mathbb{N}}
\newcommand{\F}{\mathbb{F}}

\newcommand{\lb}{[\![}
\newcommand{\rb}{]\!]}
\newcommand{\sys}[2]{\lb #1, #2\rb}
\newcommand{\innerp}[2]{\left \langle #1, #2\right \rangle}

\newcommand{\diag}{\mathrm{diag}}

\title{Neighbourhood conditions for network stability with link uncertainty
\thanks{Supported in part by the Australian Research Council (DP210103272).}}
\author{Simone Mariano and Michael Cantoni
	\thanks{S. Mariano and M. Cantoni are with the Department of Electrical and Electronic Engineering, The University of Melbourne, Australia. E-mails: {\tt \{simone.mariano,cantoni\}@unimelb.edu.au}}}

\begin{document}

\maketitle
\begin{abstract}
The main result relates to structured robust stability analysis of an input-output model for networks with link uncertainty.
It constitutes a collection of integral quadratic constraints, which together imply robust stability of the uncertain networked dynamics. Each condition is decentralized in the sense that it depends on
model data pertaining to 
the neighbourhood of 
a specific agent. By contrast, pre-existing conditions for the network model are link-wise decentralized, with each involving conservatively more localized problem data. A numerical example is presented to 
illustrate the advantage of the new broader neighbourhood conditions.  
\end{abstract}

\section{Introduction}

Motivated by problems in power and water distribution, transportation, ecology, and economics, 
large-scale networks of dynamical systems have been long studied in system and control theoretic terms; e.g., see~\cite{siljak1978large,arcak2016networks} for state-space methods, and~\cite{moylan1978stability,vidyasagar1981input} for input-output methods. 

In this paper, an input-output approach, based on 
integral quadratic constraints (IQCs)
~\cite{megretski1997system}, is pursued to 
progress the 
development of 
scalable stability certificates for networks with uncertain links. The uncertain network model considered here was recently developed in~\cite{MCLinks}, with a focus on assessing the impact of link uncertainty expressed relative to the ideal (unity gain) link.  
The work is related to~\cite{lestas2006scalable,jonsson2010scalable,andersen2014robust,khong2014scalable}, and as in \cite{MCLinks}, most closely to \cite{pates2016scalable}. The main contribution is 
an alternative approach to the decomposition of a monolithic IQC certificate that implies robust network stability. 
The new decomposition 
involves a 
collection of sufficient conditions, each depending on model data pertaining to a specific agent, its neighbours, and the corresponding links. That is, each condition is local to a specific neighbourhood. As such, compared to the link-wise decomposition 
presented in \cite{MCLinks}, each neighbourhood based condition involves more (still localized) model data. This provides scope for reduced conservativeness, as illustrated by a numerical example.

The paper is organized as follows: Various preliminaries are established next. In Section~\ref{sec:ideal+IQC}, the structured feedback model of a network with uncertain links, and a correspondingly structured IQC based robust stability condition, which is amenable to decomposition, are recalled from~\cite{MCLinks}. Then, in 
Section~\ref{sec:nodes}, 
the novel neighborhood based decomposition is developed, alongside statement of the link-wise conditions from~\cite{MCLinks} for comparison, including a numerical example. Some concluding remarks are provided in Section~\ref{sec:conc}.

\section{Preliminaries} 

\subsection{Basic notation} \label{subsec:notation}
The natural, real, and complex numbers are denoted $\naturals$, $\real$, and $\complex$, respectively.
For $i,j\in\naturals$, $[i:j]:=\{k\in\naturals~|~i\leq k \leq j\}$, which is empty if $j<i$, $\real_{\bullet \alpha}:=\{\beta\in\mathbb{R}~|~\beta \bullet \alpha\}$ for given order relation $\bullet\in\{>,\geq,<,\leq\}$, and $[\alpha,\beta]:=\mathbb{R}_{\geq \alpha} \cap \mathbb{R}_{\leq \beta}$. 

With $\F\in\{\real,\complex\}$, for $p\in\naturals$, the $p$-dimensional Euclidean space over $\F$ is denoted by $\F^p$. Given $x\in\F^p$, for $i\in[1:p]$, the scalar $x_i\in\F$ denotes the $i$-th coordinate.
The vectors $\mathbf{1}_p,\mathbf{0}_p\in\mathbb{R}^p$ satisfy $1=(\mathbf{1}_p)_i=1+(\mathbf{0}_p)_i$ for every $i\in[1:p]$.

$\F^{p\times q}$ denotes the space of $p\times q$ matrices over $\F$, for $p,q\in\naturals$.
The identity matrix is denoted by $I_p \in \mathbb{F}^{p\times p}$, the square zero matrix by $O_p\in \mathbb{F}^{p\times p}$, and the respective $p\times q$ matrices of ones and zeros by $\mathbf{1}_{p\times q}$ and $\mathbf{0}_{p\times q}$. Given $M\in \F^{p\times q}$, for $i\in[1:p]$, $j\in[1:q]$, the respective matrices $M_{(\cdot,j)} \in \F^{p\times 1}\sim\F^{p}$, and $M_{(i,\cdot)}\in\F^{1\times q}$, denote the $j$-th column, and $i$-th row. Further, $M_{(i,j)}\in\F$ denotes the entry in position $(i,j)$. Given $x\in\mathbb{F}^p$, the matrix $M=\mathrm{diag}(x)\in\mathbb{F}^{p\times p}$ is such that $M_{(i,i)}=x_i$ and $M_{(i,j)}=0$, $j\neq i\in[1:p]$, whereby $I_p=\mathrm{diag}(\mathbf{1}_p)$. The transpose of $M\in\mathbb{F}^{p\times q}$ is denoted by $M^\prime\in\mathbb{F}^{q\times p}$, and $M^*=\bar{M}^\prime$ denotes the complex conjugate transpose. For $M_i=M_i^*\in\mathbb{F}^{p\times p}$, $i\in\{1,2\}$, $M_i\succ 0$ means there exists $\epsilon>0$ such that $x^* M_i x \geq \epsilon x^*x$ for all $x\in\mathbb{F}^p$, $M_i\succeq 0$ means $x^* M_i x \geq 0$ for all $x\in\mathbb{F}^p$, and $M_1\succeq (\text{resp.}~\succ) M_2$ means $M_1-M_2 \succeq (\text{resp.}~\succ) 0$. 

\subsection{Signals and systems}
The Hilbert space of square integrable signals $v=(t\in\real_{\geq 0}\mapsto v(t)\in\real^p)$ is denoted $\L_{2\,}^p$, where the inner-product $\langle v, u \rangle := \int_{0}^{\infty}v(t)^\prime u(t)\dif t$ and norm $\|v\|_2 := \langle v, v \rangle^{1/2}$ are finite; the superscript is dropped when $p=1$. The corresponding extended space of locally square integrable signals is denoted by $\L_{2e}^p$; i.e.,  $v:\real_{\geq 0}\rightarrow\real^p$ such that $\boldsymbol{\pi}_{\tau}(v)\in\L_{2\,}^p$ for all $\tau\in\real_{\geq 0}$, where $(\boldsymbol{\pi}_\tau(v))(t):=f(t)$ for $t\in[0,\tau)$, and $(\boldsymbol{\pi}_\tau(v))(t):=0$ otherwise. 
The composition of maps $F:\L_{2e}^p \mapsto \L_{2e}^r$ and $G:\L_{2e}^q \mapsto \L_{2e}^p$ is denoted by $F\circ G := (v\mapsto F(G(v))$, and the direct sum by $F\oplus G:=((u,v)\mapsto (F(u),G(v)))$. Similarly, $\bigoplus_{i=1}^n G_i = G_1\oplus\cdots\oplus G_n$. When $G$ is linear, in the sense $(\forall \alpha,\beta\in\real)~(\forall u,v\in\L_{2e}^q)~G(\alpha u + \beta v) = \alpha G(u) + \beta G(v)$, the image of $v$ under $G$ is often written $Gv$, and in composition with another linear map $\circ$ is dropped. The action of a linear $G:\L_{2e}^{q}\rightarrow\L_{2e}^{p}$ corresponds to the action of $p\cdot q$ scalar systems $G_{(i,j)}:\L_{2e}\rightarrow \L_{2e}$, $i\in[1:p]$, $j\in[1:q]$, on the coordinates of the signal vector input associated with  $\L_{2e}^q\sim\L_{2e}\times\cdots\times\L_{2e}$; i.e., $(Gv)_i = \sum_{j=1}^q G_{(i,j)}v_j$. Matrix notation is used to denote this. Also, for convenience, the map of pointwise multiplication by a matrix on $\L_{2e}$ is not distinguished in notation from the matrix. 

A system is any map $G:\L_{2e}^q\rightarrow \L_{2e}^p$, with $G(0)=0$, that is {\em causal} in the sense  $\boldsymbol{\pi}_\tau(G(u)) = \boldsymbol{\pi}_\tau(G (\boldsymbol{\pi}_\tau(u))$ for all $\tau\in\real_{\geq 0}$. It is
called {\em stable} if $u\in\L_{2\,}^q$ implies $G(u) \in \L_{2\,}^p$ and  $\|G\|:=\sup_{0\neq u} \|G(u)\|_2/\|u\|_2< \infty$ (the composition of stable systems is therefore stable.) The feedback interconnection of $G$ with system $\Delta:\L_{2e}^p \rightarrow \L_{2e}^q$ is {\em well-posed} if for all $(d_y,d_u)\in\L_{2e}^p\times\L_{2e}^q$, there exists unique $(y,u)\in\L_{2e}^p\times\L_{2e}^q$, such that 
\begin{align} \label{eq:feedback1}
y = G(u) + d_y, \qquad u = \Delta(y) + d_u, 
\end{align}
and $[\![G,\Delta]\!] := ((d_y,d_u) \mapsto (y,u))$ is causal; see Figure~\ref{fig:realnet_0}. If, in addition, $\|[\![G,\Delta]\!]\| < \infty$, then the closed-loop is called stable. 
\begin{figure}[htbp]
\centering
\begin{tikzpicture}
  \vspace{-1.5cm}
\node [draw,
    minimum width=1.2cm,
    minimum height=1.2cm,
]  (permutation) at (0,0) {$G$};
 
\node [draw,
    minimum width=1.2cm, 
    minimum height=1.2cm, 
    below=.5cm  of permutation
]  (nets) {$\Delta$};

\node[draw,
    circle,
    minimum size=0.6cm,
    left=1cm  of nets
] (sum){};
 
\draw (sum.north east) -- (sum.south west)
(sum.north west) -- (sum.south east);
 
\node[left=-1pt] at (sum.center){\tiny $+$};
\node[above=-1pt] at (sum.center){\tiny $+$};


\node[draw,
    circle,
    minimum size=0.6cm,
  right= 1cm   of permutation 
] (sum2){};
 
\draw (sum2.north east) -- (sum2.south west)
(sum2.north west) -- (sum2.south east);
 
\node[right=-1pt] at (sum2.center){\tiny $+$};
\node[below=-1pt] at (sum2.center){\tiny $+$};

\draw[-stealth] (sum.east) -- (nets.west)
    node[midway,above]{$y$};
 
\draw[-stealth] (permutation.west) -| (sum.north);
 
\draw [stealth-] (sum.west) -- ++(-1,0) 
    node[midway,above]{$d_{y}$};

\draw[-stealth]  (nets.east) -|  (sum2.south);
 
\draw[-stealth] (sum2.west) -- (permutation.east) 
    node[midway,above]{$u$};
 
\draw [stealth-] (sum2.east) -- ++(1,0) 
    node[midway,above]{$d_{u}$};
 \vspace{-1cm}
\end{tikzpicture}
    \caption{Standard feedback interconnection.  \vspace{-0.5cm}}
    \label{fig:realnet_0}      
\end{figure}
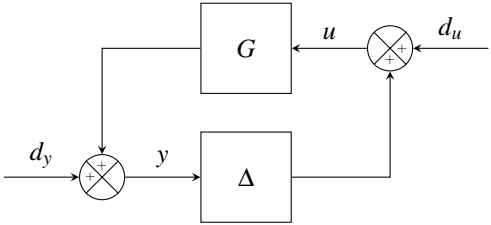

The following result is the well-known IQC robust feedback stability theorem, taken from~\cite{megretski1997system}:
\begin{theorem}
\label{thm:robust_stability}
Given stable system $\Delta:\L_{2e}^p \rightarrow \L_{2e}^q$, and bounded linear map $\Pi:\L_{2\,}^p\times\L_{2\,}^q \rightarrow \L_{2\,}^p\times \L_{2\,}^q$ that is self-adjoint in the sense $(\forall g_1,g_2\in\L_{2\,}^p\times \L_{2\,}^q)~\langle g_1, \Pi g_2 \rangle=\langle \Pi g_1, g_2\rangle$, suppose 
\begin{equation}
\label{eq:stab_IQC_Delta}
\left\langle 
(y,u),
\Pi 
(y,u)
\right\rangle \geq 0,
\quad
u=\alpha \Delta(y),
\end{equation}
for all $(y,\alpha)\in\L_{2\,}^p \times [0,1]$. Further, given stable system $G:\L_{2e}^q\rightarrow \L_{2e}^p$, suppose $[\![G,\alpha \Delta]\!]$ is well-posed for all $\alpha\in[0,1]$, and there exists $\epsilon>0$ such that 
\begin{equation}
\label{eq:stab_G}
\left\langle 
(y,u),
\Pi 
(y,u)
\right\rangle \leq -\epsilon \| u\|_2^2,
\quad
y = G(u),
\end{equation}
for all $u\in\L_{2\,}^q$.
Then, $\sys{G}{\Delta}$ is stable. 
\end{theorem}

\subsection{Graphs} \label{sec:graphs}
Let $\G=(\V,\E)$ be a simple (self-loopless and undirected) graph, where $\V=[1:n]$ is the set of $n\in\mathbb{N}\setminus\{1\}$ vertices, and $\E\subset\{\{i,j\}~|~i,j\in\V\}$ is the set of $m:=|\E|\in\mathbb{N}$ edges. Bijective $\kappa_{\E}:\E\rightarrow \M$ with $\M:=[1:m]$ denotes a fixed enumeration of the edge set $\E$ for indexing. The set $\N_i:=\{j~|~\{i,j\}\in\E\}$ comprises the neighbours of $i\in\V$, $\E_i:=\{\{i,j\}~|~j\in\mathcal{N}_i\}$ is the corresponding neighbourhood edge set, 
and bijective $\kappa_{\E_i}:\E_i\rightarrow \M_i$ is a fixed enumeration of the $m_i:=|\E_i|$ edges, where $\M_i:=[1:m_i]$. The edge indexes associated with $i\in \V$ are gathered in the set denoted by $\K_i:=\{\kappa_{\E}(\{i,j\}) \mid j \in \N_i\}$. For each $k\in \M$, the set $\Ls_k:=(\K_i \cup \K_j) \setminus \{k\}$, where $\{i,j\} = \kappa_{\E}^{-1}(k)$, is the collection of all indexes of the edges associated with either of the neighbouring vertices $i$ and $j$, excluding the one linking them. The following technical result regarding the given graph $\mathcal{G}=(\V,\E)$ is used subsequently.
\begin{lemma}
    \label{lem:graph_sums_equivalence}
    Given arbitrary $E_k,F_k:\L_{2\,}^p \rightarrow \L_{2\,}^p$ 
    for $k\in \M$: 
    \begin{enumerate} [i)]
        \item \label{lem:graph_sums_equivalence1} $\displaystyle \sum_{i\in \V}\sum_{k\in \K_i} \tfrac{1}{2}F_k= \sum_{k\in \M} F_k$;
        \item \label{lem:graph_sums_equivalence4} $\displaystyle  \sum_{i\in \V}\sum_{k\in \K_i} \,\sum_{\ell\in \K_i\setminus\{k\}} E_k \circ F_\ell = \sum_{k\in \M} \sum_{\ell \in \Ls_k} E_k \circ F_\ell$.
    \end{enumerate}
\end{lemma}
\begin{proof}
See Appendix. 
\end{proof}

\section{Networked system model and robust stability analysis}
\label{sec:ideal+IQC}

In this section, the network model and robust stability analysis from \cite{MCLinks} is recalled first. A new result is then derived to underpin the aforementioned neighbourhood decomposition, which is subsequently developed in Section~\ref{sec:nodes}. 

Consider a network of $n\in \naturals\setminus\{1\}$ dynamic agents, coupled according to the simple graph $\G=(\V,\E)$. The vertex set $\V=[1:n]$ corresponds to a fixed enumeration of the agents, and  
$m:=|\E|$ is the number of edges, defined according to $\{i,j\}\in\E$ if the output of agent $i\in\V$ is shared as an input to agent $j\in\V$, and vice-versa. It is assumed that the number of neighbours $m_i:=|\N_i|\geq 1$ for all $i\in\V$. Note that $\sum_{i\in \V} m_i=2m$. To tame the notation, each agent has a single output and dynamics corresponding to the system $H_i:\L_{2e}^{m_i}\rightarrow \L_{2e}$, which is taken to be linear and stable, with the vector input signal coordinate order fixed by the neighbourhood edge-set enumerations $\kappa_{\E_i}:\E_i\rightarrow\M_i$.

Define the block diagonal systems
\begin{subequations}
\label{eq:blkdiag}
\begin{align}
&H:=\bigoplus_{i=1}^{n}H_i:\L_{2e}^{2m}\rightarrow \L_{2e}^{n}, 
\quad 
T:=\bigoplus_{i=1}^{n} \mathbf{1}_{m_{i}\times 1}:\L_{2e}^{n}\rightarrow \L_{2e}^{2m}, \label{eq:Tdef}\\ 
&\text{ and } \quad R:=\bigoplus_{i=1}^{n} \big(\bigoplus_{k=1}^{m_i} R_{i,k}\big):\L_{2e}^{2m}\rightarrow\L_{2e}^{2m},
\end{align}
\end{subequations}
 where $\mathbf{1}_{m_{i}\times 1}:\L_{2e}\rightarrow\L_{2e}^{m_i}$ denotes pointwise multiplication by $\mathbf{1}_{m_{i}\times 1}\in\mathbb{R}^{m_{i}\times 1}$, and the system $R_{i,k}:\L_{2e}\rightarrow \L_{2e}$ represents the stable but possibly nonlinear and time-varying dynamics of the link {\em from} agent $i\in\V$ {\em to} its neighbour $j\in\N_i$ with $\{i,j\}=\kappa_{\E_i}^{-1}(k)$.  Given these components, the networked system can be modelled as the structured feedback interconnection $[\![P, R \circ T \circ H]\!]$, where $P:\L_{2e}^{2m}\rightarrow \L_{2e}^{2m}$ is pointwise multiplication by a permutation matrix arising from the structure of $\G$. More specifically, for each $i\in\V$, $k\in\M_i$, and $r\in[1:2m]$, the corresponding entry of this permutation matrix is given by
 \begin{align} \label{eq:permute}
 P_{(\sum_{h\in [1:i-1]} m_h
 +k,\,r)}\!=\!
 \begin{cases} 1 & \text{if}~\!r\!=\!
 \sum_{h\in [1:j-1]} \!m_h 
 \!+\!\kappa_{\E_j}(\!\{i,j\}\!)
 \\ &\text{with } j\in\kappa_{\E_i}^{-1}(k)\setminus \{i\},\\
 0 & \text{otherwise},
 \end{cases}
 \end{align}
 where by convention the sum over an empty index set is zero.
 See \cite{MCLinks} for more details about the model and its components.
 \begin{figure*}[t]
\hspace{0pt} 
\centering
\begin{minipage}{.43\linewidth}
\centering
\hspace{-5pt} 
\begin{tikzpicture}
 
\node [draw,
    minimum width=1.2cm,
    minimum height=1.2cm,
]  (permutation) at (0,0) {$P$};
 
\node [draw,
    minimum width=2cm, 
    minimum height=1.2cm, 
    below=.5cm  of permutation
]  (nets) {$R\circ T \circ H$};

\node[draw,
    circle,
    minimum size=0.6cm,
    left=0.8cm  of nets
] (sum){};
 
\draw (sum.north east) -- (sum.south west)
(sum.north west) -- (sum.south east);
 
\node[left=-1pt] at (sum.center){\tiny $+$};
\node[above=-1pt] at (sum.center){\tiny $+$};


\node[draw,
    circle,
    minimum size=0.6cm,
  right= 1.2cm   of permutation 
] (sum2){};
 
\draw (sum2.north east) -- (sum2.south west)
(sum2.north west) -- (sum2.south east);
 
\node[right=-1pt] at (sum2.center){\tiny $+$};
\node[below=-1pt] at (sum2.center){\tiny $+$};

\draw[-stealth] (sum.east) -- (nets.west)
    node[midway,above]{$v$};
 
\draw[-stealth] (permutation.west) -| (sum.north);
 
\draw [stealth-] (sum.west) -- ++(-0.5,0) 
    node[midway,above]{$d_{v}$};

\draw[-stealth]  (nets.east) -|  (sum2.south);
 
\draw[-stealth] (sum2.west) -- (permutation.east) 
    node[midway,above]{$w$};
 
\draw [stealth-] (sum2.east) -- ++(0.5,0) 
    node[midway,above]{$d_{w}$}; 
 \end{tikzpicture}
 \! \vspace{40 pt} \, 
\hspace*{1.25cm}
\begin{tikzpicture}
 
\node [draw,
    minimum width=1.4cm,
    minimum height=1.2cm,
     below=.5cm  of nets
]  (permutation2) {$H \circ P$};
 
\node [draw,
    minimum width=1.4cm, 
    minimum height=1.2cm, 
    below=.5cm  of permutation2
]  (nets) {$R\circ T$};

\node[left=1cm  of nets] (sum){};


\node[draw,
    circle,
    minimum size=0.6cm,
  right= 1cm   of permutation2 
] (sum2){};
 
\draw (sum2.north east) -- (sum2.south west)
(sum2.north west) -- (sum2.south east);
 
\node[right=-1pt] at (sum2.center){\tiny $+$};
\node[below=-1pt] at (sum2.center){\tiny $+$};

\draw[-stealth] (sum.center) -- (nets.west)
    node[midway,above]{$\tilde{v}$};
 
\draw(permutation2.west) -| (sum.center);

\draw[-stealth]  (nets.east) -|  (sum2.south);
 
\draw[-stealth] (sum2.west) -- (permutation2.east) 
    node[midway,above]{$\tilde{w}$};
 
\draw [stealth-] (sum2.east) -- ++(0.5,0) 
    node[midway,above]{$\qquad \qquad d_w+P^{-1}d_v$};
\end{tikzpicture} 
\end{minipage}
\hspace{3pt} 
\begin{minipage}{.43\linewidth}
\centering
\resizebox{6.03cm}{7.45cm}{\begin{tikzpicture}
 
\node [draw,
    minimum width=1.4cm,
    minimum height=1.2cm,
]  (permutation) at (0,0) {$H \circ P$};

\node[below=2.5cm  of permutation] (middle){};

\node [draw,
    minimum width=1.4cm, 
    minimum height=1.2cm, 
    below=5cm  of permutation
]  (nets) {$R\circ T$};

\node [draw,
    minimum width=1.0cm, 
    minimum height=1.0cm, 
    below=.3cm  of permutation
]  (T1) {$T$};

\node [draw,
    minimum width=1.0cm, 
    minimum height=1.0cm, 
    above=.3cm  of nets
]  (T2) {$T$};

\node[left=1cm  of nets] (n1){};

\node[above=1cm  of nets] (n2){};

\node[left=1cm  of n2] (n3){};

\node[left=1cm  of permutation] (n4){};

\node[right=1cm  of permutation] (n5){};

\node[left=1.2cm  of T1] (n6){};

\node[left=1.2cm  of T2] (n7){};

\node[right=1.2cm  of T1] (n8){};

\node[right=1.2cm  of T2] (n9){};

\node[left=1cm  of nets] (n10){};

\node [ 
    above=0.8cm  of n10
]  (S) {};

\node [
    below=0.8cm  of n4
]  (Sda) {};


\node[draw,
    circle,
    minimum size=0.6cm,
  right=1.0cm of T2
] (sum3){};
 
\draw (sum3.north east) -- (sum3.south west)
(sum3.north west) -- (sum3.south east);
 
\node[left=-1pt] at (sum3.center){\tiny $-$};
\node[below=-1pt] at (sum3.center){\tiny $+$};


\node[draw,
    circle,
    minimum size=0.6cm,
above=1. cm of sum3
] (sum2){};
 
\draw (sum2.north east) -- (sum2.south west)
(sum2.north west) -- (sum2.south east);
 
\node[right=-1pt] at (sum2.center){\tiny $+$};
\node[below=-1pt] at (sum2.center){\tiny $+$};


\node[draw,
    circle,
    minimum size=0.6cm,
  right=1.0cm of T1
] (sum4){};
 
\draw (sum4.north east) -- (sum4.south west)
(sum4.north west) -- (sum4.south east);
 
\node[left=-1pt] at (sum4.center){\tiny $+$};
\node[below=-1pt] at (sum4.center){\tiny $+$};

\draw[-stealth] (n1.center) -- (nets.west)
    node[midway,above]{$\tilde{v}$};
 
\draw(permutation.west) -| (Sda.center);

\draw(Sda.center) -| (S.center);

\draw(S.center) -- (n1.center);

\node[above=1.6cm of S] (ttS) {}; 
\node[left=-0.1 of ttS] {$y$};
\node[above=0.2cm of sum2] (ttsum) {};
\node[right=-0.1cm of ttsum] {$u$};

\draw[-stealth]  (nets.east) -|  (sum3.south);
 
\draw[-stealth] (n5.center) -- (permutation.east) 
    node[midway,above]{$\tilde{w}$};
    
\draw[-stealth] (n6.center) -- (T1.west);
\draw[-stealth] (n7.center) -- (T2.west);

\draw[stealth-] (sum4.west) -- (T1.east);
\draw[stealth-] (sum3.west) -- (T2.east);

\draw[-stealth] (sum2.north) -- (sum4.south);
 
\draw[-stealth] (sum3.north) -- (sum2.south);

\draw (n5.center) -- (sum4.north) 
    node[midway,above]{};
 
\draw [stealth-] (sum2.east) -- ++(0.5,0) 
    node[right,above]{$d$};

 \node[draw,dashed,inner sep=7pt, yscale=1, fit={(T1) (Sda) (permutation) (sum4)}] (G) {};

 \node[draw,dashed,inner sep=7pt, yscale=1, fit={(T2) (S) (nets) (sum3)}] (D) {};

\node[left=0.25cm of D] (D1) {$\Delta$};
\node[left=0.25cm of G] (G1) {$G$};
\node[midway,above]{$\hspace{-2.5 cm}{\vspace{0.2cm}\tilde{v}}$};
 
 
\end{tikzpicture}}
\end{minipage}
    \caption{Networked system model $[\![P,R\circ T\circ H]\!]$, and loop transformations for robust stability analysis.}
    \label{fig:transfloop}      
\end{figure*}
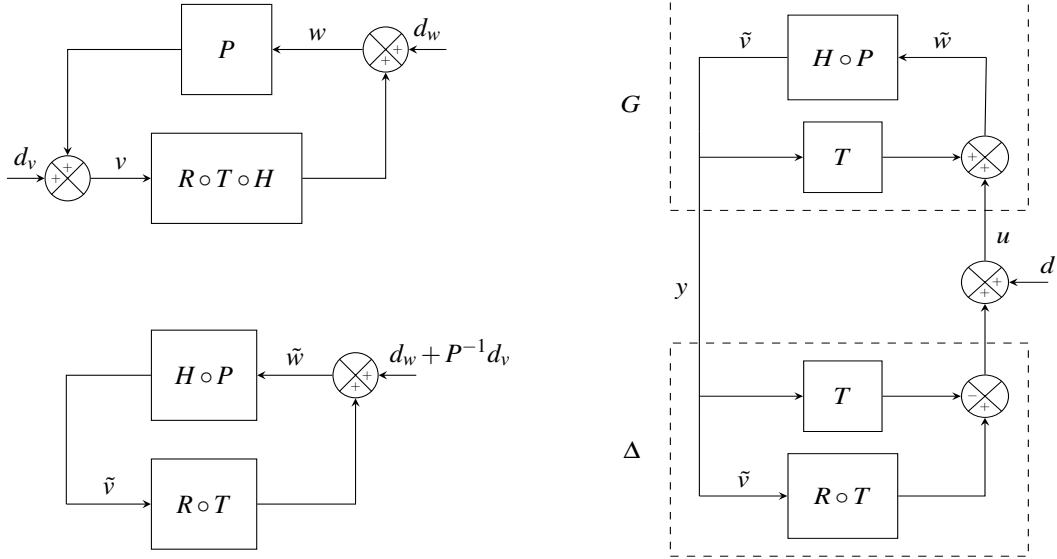

The `routing' matrix $P=P^\prime=P^{-1}$ is the adjacency matrix of an undirected $1$-regular sub-system graph 
 $\widetilde{\G}:=(\widetilde{\V},\widetilde{\E})$, with $2m$ vertices and $m$ edges, corresponding to the disjoint union of the two-vertex subgraphs $\mathcal{G}[e]$ induced by each edge $e\in\E$ in the network graph $\mathcal{G}=(\V,\E)$. As such, 
\begin{equation}
 \label{eq:Padjacency}
  P=D-L=I_{2m}-
 \sum_{k\in \M} L_k,
 \end{equation}
where the degree matrix $D=I_{2m}$ since $\widetilde{\G}$ is 1-regular, and the Laplacian decomposes as
 \begin{equation}
 \label{eq:Laplacian_Subsystem_Graph}
 L=\sum_{k\in \M}L_k
  =\sum_{k\in \M}B_{(\cdot,k)}B_{(\cdot,k)}^\prime
  =B B^\prime,   
 \end{equation}
 where $B\in\mathbb{R}^{2m\times m}$ is the incidence matrix
 defined by 
$B_{(r,k)}=1=-B_{(s,k)}$ for $\{r,s\}=\kappa^{-1}_{\widetilde{\E}}(k)$, and $B_{(t,k)}:=0$ for each $t \in [1:2m]\setminus\{r,s\}$, over $k\in\M$; the enumeration $\kappa_{\widetilde{\E}}$ is taken to be compatible with the enumerations $\kappa_{\E_i}$ and the definition of $P$ in \eqref{eq:permute}. The edge orientation is arbitrary.
 \begin{assumption}
\label{ass:coprime} 
 The network with ideal links (i.e., $R=I_{2m}$) is stable in the sense that $[\![P, T\circ H]\!]$ 
 is stable.
\end{assumption}
As detailed in \cite{MCLinks}, one can leverage Assumption~\ref{ass:coprime} in the analysis of $[\![P,R\circ T \circ H]\!]$ by considering uncertainty in the links $R$ relative to ideal unity gain links.
Indeed, as illustrated in Figure~\ref{fig:transfloop}, to verify robust stability of the networked system $[\![P,R\circ T \circ H]\!]$ it is sufficient to verify that $[\![G,\Delta]\!]$ is stable, where  
\begin{equation}
\label{eq:Delta}
    \Delta:= (R-I_{2m}) \circ T,
\end{equation}
and
\begin{equation}
\label{eq:G}
G := 
H \circ (P - T\circ H )^{-1},
\end{equation}
with the stable systems $R$, $H$, and $T$ as per \eqref{eq:blkdiag}, and $P$ as per \eqref{eq:Padjacency}. Since $(P - T\circ H )^{-1}$ is stable by Lemma~1 in \cite{MCLinks}, $G:\L_{2e}^{2m}\rightarrow \L_{2e}^{n}$, and  the uncertain $\Delta:\L_{2e}^{n}\rightarrow \L_{2e}^{2m}$, are both stable. 
\begin{theorem} (\!\!\cite[Thm.~2]{MCLinks})
\label{thm:stable_sampledH} 
Under Assumption~\ref{ass:coprime}, if 
$[\![G,\Delta]\!]$ is stable, with $G$ as per \eqref{eq:G}, and $\Delta$ as per \eqref{eq:Delta}, then the networked system model $[\![P,R\circ T \circ H]\!]$ is stable. Further, when $\Delta$ is also linear, stability of $[\![P,R\circ T \circ H]\!]$ implies stability of $[\![G,\Delta]\!]$.  
\end{theorem}
Suppose, for all $(y,\alpha)\in \L_{2\,}^{n}\times [0,1]$,
\begin{equation} \label{eq:DeltaIQC}
\innerp{(y,u)}{\Phi(y,u)} \geq 0, \quad u=\alpha \Delta(y),
\end{equation}
where the given bounded linear self-adjoint operator $\Phi:(\L_{2\,}^{n} \times \L_{2\,}^{2m}) \rightarrow (\L_{2\,}^{n} \times \L_{2\,}^{2m})$ is structured according to
\begin{align}
\label{eq:multip}
    \Phi = \begin{bmatrix} \Phi_1 & \Phi_2 \\ \Phi_2^* & \Phi_3 \end{bmatrix}
    =\begin{bmatrix} \bigoplus_{i=1}^n \Phi_{1,i} &  \bigoplus_{i=1}^n \Phi_{2,i} \\ \bigoplus_{i=1}^n\Phi_{2,i}^* & \bigoplus_{i=1}^n \Phi_{3,i} \end{bmatrix};
\end{align}
the superscript $*$ denotes Hilbert adjoint. For example, this IQC holds by selecting each
$$\Phi_i := 
\begin{bmatrix} \Phi_{1,i} & \Phi_{2,i} \\ \Phi_{2,i}^*
 & \Phi_{3,i} \end{bmatrix}
:\L_{2\,}\times\L_{2\,}^{m_i} \rightarrow 
\L_{2\,} \times \L_{2\,}^{m_i},\quad i\in\V,$$
such that $\innerp{(y_i,u_i)}{\Phi_i(y_i,u_i)} \geq 0$, $u_i=\alpha\Delta_i(y_i)$, for all $(y_i,\alpha)\in\L_{2}\times [0,1]$, with local $\Delta_i=(\oplus_{k=1}^{m_i}R_{i,k}-I_{m_i})\circ \mathbf{1}_{m_i\times 1}$. Then, by Theorem~\ref{thm:robust_stability}, the stability of $\sys{G}{\Delta}$ is implied by the existence of $\epsilon>0$ such that for all $z\in \L_{2\,}^{2m}$,
\begin{align}
\label{eq:stabiliy_G_2}
        \innerp{\begin{bmatrix}
	  N \\
		M 
        \end{bmatrix}z}{\begin{bmatrix} \Phi_1 & \Phi_2 \\ \Phi_2^* & \Phi_3 \end{bmatrix}  
        \begin{bmatrix}
	  N \\
		M 
        \end{bmatrix}z}
         \leq -\epsilon \| z \|_2^2,
\end{align} 
where
\begin{gather}
    \label{eq:coprimeNM}
    N:= H, \quad \mathrm{and} \quad    M:=(P-T\circ H)=I_{2m}-L-T\circ H,
\end{gather}
with $T$, $H$, and $P$, as per \eqref{eq:blkdiag} and \eqref{eq:permute}.
Importantly, $N$ and $M$ are structured coprime factors of $G=NM^{-1}$, 
as discussed in \cite[Rem.~4]{MCLinks}. Also refer to \cite[Sec.~IV]{MCLinks} for more on the use of Theorem~\ref{thm:robust_stability} to arrive at \eqref{eq:stabiliy_G_2}, and discussion of the obstacle to direct application of the decomposition method proposed in~\cite{pates2016scalable} to the 
equivalent condition
\begin{align}
  \innerp{\begin{bmatrix}
	  I_{2m}\\
		L \\
        \end{bmatrix}z}{\left[ \begin{array}{cc} \Xi_1 & \Xi_2 \\
    	\Xi_2^{*}& \Xi_3\end{array} \right] 
        \begin{bmatrix}
	  I_{2m}\\
		L \\
        \end{bmatrix}z} \leq -\epsilon \| z \|_2^2,
       \label{eq:stabiliy_G_3b}
\end{align}
where $L=\sum_{k\in [1:m]} L_k$ is the sub-system graph Laplacian in~\eqref{eq:Laplacian_Subsystem_Graph},
\begin{subequations}
    \label{t4Pi}
    \begin{align}
    \Xi_{1}&:= N^*\Phi_{1}N+ N^*\Phi_2 J+J^*\Phi_2^*N+J^*\Phi_3 J~,\\
    \Xi_{2}&:=-N^*\Phi_2-J^*\Phi_3~,\\
    \Xi_{3}&:=\Phi_3~,
\end{align}
and $J:= I_{2m}-T\circ H$.
\end{subequations}
The block-diagonal structure of $N$, $J$, $\Phi_1$, $\Phi_2$, $\Phi_3$, and the Hilbert adjoints when restricted to $\L_{2\,}$, is such that $\Xi_1 =  \bigoplus_{i=1}^n \Xi_{1,i}$, $\Xi_{2}= \bigoplus_{i=1}^n \Xi_{2,i}$, and $\Xi_{3} =\bigoplus_{i=1}^n \Xi_{3,i}$, where 
\begin{align*}
    \Xi_{1,i} &:= H_i^*\Phi_{1,i}H_i +(I_{m_i} - \mathbf{1}_{m_i}H_i)^*\Phi_{3,i} (I_{m_i} - \mathbf{1}_{m_i}H_i) \\
    &\quad + H_i^*\Phi_{2,i}(I_{m_i} - \mathbf{1}_{m_i}H_i) +(I_{m_i} - \mathbf{1}_{m_i}H_i)^*\Phi_{2,i}^*H_i~,\\
    \Xi_{2,i}&:=-H_i^*\Phi_{2,i}-(I_{m_i}-\mathbf{1}_{m_i}H_i)^*\Phi_{3,i}~,\\
    \Xi_{3,i}&:=\Phi_{3,i}~.
\end{align*}
For context, an existing
link-wise
decomposition of \eqref{eq:stabiliy_G_3b} is
first recalled from~\cite{MCLinks}.
\begin{lemma} (\!\!\cite[Lem.~3]{MCLinks}) \label{lem:PatesReadapted} 
Let $W=\sum_{k\in \M} W_k \in \mathbb{R}^{2m\times 2m}$ be such that $W_k\succeq 0$ and $W\succ 0$.
Suppose there exist $X_k\!=\!X^*_k:\L_{2\,}^{2m}\rightarrow\L_{2\,}^{2m}$,  $Z_k\!=\!Z_k^*:\L_{2\,}^{2m}\rightarrow\L_{2\,}^{2m}$, and $\epsilon_k>0$, $k\in \M$, such that for all $z\in\L_{2\,}^{2m}$,
\begin{align}
    \label{ISCk}
       & \innerp{\begin{bmatrix}
		I_{2m} \\
		L_k\\
		\end{bmatrix} z} {\begin{bmatrix}
		X_k+\epsilon_k W_k & \Xi_{2} \\
    	\Xi_{2}^{*} & Z_k \end{bmatrix}
     \begin{bmatrix}
		I_{2m} \\
		L_k\\
		\end{bmatrix}z}\leq 0,~ k\in\M,\\
 \label{ISCCompa}
    &       \innerp{z}{\Xi_1 z} - \sum_{k\in \M} \innerp{z}{X_k z} \leq 0,\\
\label{ISCCompb}
      &     \innerp{L z}{\Xi_3 L z} - \sum_{k\in \M} \innerp{L_k z}{Z_k L_k z} \leq 0,
\end{align}  
with $L_k$ as per~\eqref{eq:Laplacian_Subsystem_Graph}.
 Then, there exists $\epsilon>0$ such that \eqref{eq:stabiliy_G_3b} for all $z\in\L_{2\,}^{2m}$.
 \end{lemma}
The $m=|\E|$ conditions in \eqref{ISCk} 
enable use of the structure of each $L_k=B_{(\cdot,k)}B^\prime_{(\cdot,k)}$ for link-wise decentralized verification of \eqref{eq:stabiliy_G_3b}, although this can be conservative. 
An alternative is to decompose according to the broader structure of each
 \begin{equation}
 \label{eq:collection_Laplacian_nodes}
 K_i=\sum_{k\in\K_i} L_k, \quad i\in \V,
 \end{equation}
 which encompasses the neighborhood of agent $i$.
\begin{lemma}\label{lem:PVRNodes}
Let $W=\sum_{i\in \V} W_i \in \mathbb{R}^{2m\times 2m}$ be such that $W_i\succeq 0$ and $W\succ 0$.
Suppose there exist $X_i\!=\!X^*_i:\L_{2\,}^{2m}\rightarrow\L_{2\,}^{2m}$, $Y_i:\L_{2}^{2m}\rightarrow\L_{2}^{2m}$, $Z_i\!=\!Z_i^*:\L_{2\,}^{2m}\rightarrow\L_{2\,}^{2m}$, and $\epsilon_i>0$, $i\in \V$, such that for all $z\in\L_{2\,}^{2m}$,
\begin{subequations}
\label{eq:node_generic}
\begin{align}
    \label{eq:PVRNodes_IPk}
       & \innerp{\begin{bmatrix}
		I_{2m} \\
		K_i\\
		\end{bmatrix} z} {\begin{bmatrix}
		X_i+\epsilon_i W_i & Y_{i} \\
    	Y_{i}^{*} & Z_i \end{bmatrix}
     \begin{bmatrix}
		I_{2m} \\
		K_i\\
		\end{bmatrix}z}\leq 0,~~ i\in\V,\\
 \label{eq:PVRGenNodes_IPX}
      &     \innerp{z}{\Xi_1 z} - \sum_{i\in \V}\innerp{z}{X_i z} \leq 0,\\
\label{eq:PVRGenNodes_IPY}
      &     \innerp{z}{\Xi_2 L z} - \sum_{i\in \V}\innerp{z}{Y_i K_i z} \leq 0,\\    
\label{eq:PVRGenNodes_IPZ}
      &     \innerp{L z}{\Xi_3 L z} - \sum_{i\in \V}\innerp{K_i z}{Z_i K_i z} \leq 0,
\end{align}   
\end{subequations}
with $K_i$ as per \eqref{eq:collection_Laplacian_nodes}.
 Then, there exists $\epsilon\in\mathbb{R}_{>0}$ such that \eqref{eq:stabiliy_G_3b} for all $z\in\L_{2\,}^{2m}$.
 \end{lemma}
 \begin{proof}
	 For all $z\in\L_{2}^{2m}$,   \eqref{eq:PVRNodes_IPk} implies
	    \begin{align*}
	        \innerp{z}{X_i z} 
         + \innerp{z}{Y_i K_i z} + \innerp{K_i z}{Y_i^* z} 
         + \innerp{K_i z}{Z_i K_i z} \leq -\epsilon_i \innerp{z}{W_i z},
	    \end{align*}
     for each $i\in \V$, and therefore, 
	    \begin{align}
	        \sum_{i\in \V}   \Bigg ( \innerp{z}{X_i z}  &+ \innerp{z}{Y_i  K_i z} +  \innerp{K_i z}{Y_i^* z} + \innerp{K_i z}{Z_i K_i z} \Bigg ) \nonumber \\ &\leq -\sum_{i\in \V} \epsilon_i \innerp{z}{W_i z} \leq - \epsilon || z ||_2^2, \label{eq:PVRGenNodes_Proof1}
	    \end{align}
      where $\epsilon = \big(\min_{i\in \V} \epsilon_i\big) \cdot \big(\min_{x\in\mathbb{R}^{n}} x^\prime W x/x^\prime x\big)>0$; note that $W=\sum_{i\in \V} W_i \succ 0$ implies $\sum_{i\in \V} \epsilon_i W_i \succeq (\min_{i\in \V}\epsilon_i) \sum_{i\in \V}  W_i\succ 0$, because each $W_i\succeq 0$.
	Combining \eqref{eq:PVRGenNodes_IPX}, \eqref{eq:PVRGenNodes_IPY}, \eqref{eq:PVRGenNodes_IPZ}, and \eqref{eq:PVRGenNodes_Proof1}, gives 
     \begin{equation}\label{ProofISC12}	\innerp{z}{\Xi_1 z} + \innerp{z}{\Xi_2 L z}  + \innerp{L z}{\Xi_2^*  z} + \innerp{L z}{\Xi_3 L z}  \leq -\epsilon || z ||_2^2,
 \end{equation}
 as claimed.
 \end{proof}
As with~\cite[Lem. 1]{MCLinks}, the proof of Lemma~\ref{lem:PVRNodes} expands upon ideas from the proof of~\cite[Thm.~1]{pates2016scalable}, 
which is not directly applicable here for the reasons elaborated in~\cite[Rem.~4]{MCLinks}. Considering $K_i$ as per Lemma~\ref{lem:PVRNodes}, instead of $L_k$ as per Lemma~\ref{lem:PatesReadapted}, makes each instance of \eqref{eq:PVRNodes_IPk} depend on more network model data, which provides scope for reducing the conservativeness of Lemma~\ref{lem:PatesReadapted}. This comes at the cost of the additional inequality $\eqref{eq:PVRGenNodes_IPY}$,
which has no counterpart in Lemma~\ref{lem:PatesReadapted} since $\sum_{k\in\M} L_k = L$, whereas $\sum_{i\in \V} K_i \neq L$. 

\section{Main result: Neighbourhood conditions}
\label{sec:nodes}
A possible selection of $W_i$, $X_i$, $Y_i$, $Z_i$, $i\in \V$, is devised for Lemma~\ref{lem:PVRNodes}. It yields a decentralized collection of conditions 
that together imply the stability of $\sys{G}{\Delta}$, 
and thus, stability of the network by Theorem~\ref{thm:stable_sampledH}.
The conditions are 
decentralized 
in the sense that each depends on model data that is local to the neighbourhood of a specific agent, including corresponding components of the IQC based uncertainty description of the local links.

The subsequent matrix definitions, and related properties, lead to the proposed selection of $W_i$, $X_i$ $Y_i$ and $Z_i$ in Lemma~\ref{lem:PVRNodes}. The definitions pertain to the structure of the networked system model, encoded by the network graph $\G=(\V,\E)$ and corresponding sub-system graph $\widetilde{\G}=(\widetilde{\V},\widetilde{\E})$, with $\V=[1:n]$, $m=|\E|=|\widetilde{\E}|$,
$|\widetilde{\V}|=2m=\sum_{i\in\V}m_i$, and $m_i=|\N_i|$, as per Section~\ref{sec:ideal+IQC}. First, for each $k\in\M=[1:m]$, define
\begin{align} \label{eq:whatB}
    A_k := (\mathrm{diag}(B_{(\cdot,k)}))^2 \in \mathbb{R}^{2m\times 2m}, 
\end{align}
where $B\in\mathbb{R}^{2m\times 2m}$ is the incidence matrix of the sub-system graph Laplacian matrix $L=\sum_{k\in\M} L_k = \sum_{k\in \M} B_{(\cdot,k)}B_{(\cdot,k)}^\prime$ in \eqref{eq:Laplacian_Subsystem_Graph}. The matrix $A_k$ is diagonal, with $\{0,1\}$ entries, and 
$$(A_k)_{(r,r)}=1 ~\text{ if and only if } ~ r\in\kappa^{-1}_{\widetilde{\E}}(k);$$
i.e., the value is $1$ only in the two locations corresponding to the sub-systems in $\widetilde{\V}$ associated with link $k$,
as per the
definition of the incidence matrix $B$ below \eqref{eq:Laplacian_Subsystem_Graph}. Indeed, since $\widetilde{\G}=(\widetilde{\V},\widetilde{\E})$ is $1$-regular, direct calculation gives $A_kB_{(\cdot,k)}=B_{(\cdot,k)}$, $B_{(\cdot,k)}^\prime A_k = B_{(\cdot,k)}^\prime$, $A_k B_{(\cdot,\ell)}=\mathbf{0}_{2m}$, $ B_{(\cdot,\ell)}^\prime A_k=\mathbf{0}_{2m}^\prime$,
\begin{subequations}
    \label{eq:BhatMagic}
    \begin{align}
        A_k L_k = (\mathrm{diag}(B_{(\cdot,k)}))^2 B_{(\cdot,k)} B_{(\cdot,k)}^\prime = 
        L_k = L_k^\prime = L_k A_k \label{eq:BhatmagicA}\\
       \text{ and }~  A_k L_\ell = (\mathrm{diag}(B_{(\cdot,k)}))^2 B_{(\cdot,\ell)} B_{(\cdot,\ell)}^\prime = O_{2m} = L_\ell A_k \label{eq:BhatMagicB}
    \end{align}
    for all $k \neq \ell \in \M$. As such, given $i\in\V$, for $k\in\K_i$,
    \begin{align}
        (\sum_{\ell\in\K_i} L_\ell) \,A_k &= L_k =
        A_k (\sum_{\ell\in\K_i} L_\ell),  \text{ and } \label{eq:BhatMagicC}\\
     ( \sum_{k\in\K_i} A_k ) ( \sum_{\ell\in\K_i} L_\ell ) &= \sum_{k\in\K_i} L_k 
        = ( \sum_{\ell\in\K_i} L_\ell )  ( \sum_{k\in\K_i} A_k ),
        \label{eq:BhatMagicD} 
    \end{align}
     where the set $\K_i$ of edge indexes associated with agent $i$ is defined as per Section~\ref{sec:graphs}.
\end{subequations}
Finally, 
for each $i\in\V$, define 
 \begin{align} \label{eq:C}
    C_i:=\diag(T_{(\cdot,i)}) \in \mathbb{R}^{2m\times 2m},
 \end{align}
where 
$T$ is given in \eqref{eq:Tdef}. As such,
$C_i=C_i^\prime$ is a diagonal matrix with $\{0,1\}$ entries. 
Composing it with $\Xi_1$ in \eqref{t4Pi}
isolates only the model data related to the agent $i\in\V$.
It can be shown by direct calculation that 
 \begin{align}
 \label{eq:propC}
  C_i \Xi_1= \Xi_1 C_i, 
  \quad 
  \text{ and }
  \quad 
  \sum_{i \in \V} C_i=I_{2m}.
 \end{align}
The following is used to prove the subsequent main result.
\begin{lemma}
\label{lem:prop_L}
    For arbitrary linear $\varGamma = \oplus_{i=1}^n \varGamma_i$, $\varGamma_i:\L_2^{m_i}\rightarrow\L_2^{m_i}$,
    and the sub-system graph Laplacian $L=\sum_{k\in\M} L_k$ in~\eqref{eq:Laplacian_Subsystem_Graph}, 
    \begin{equation*}
       L \, \varGamma L = \sum_{k\in\M} L_k\,  \varGamma L_k + \sum_{k\in\M}  \sum_{\ell \in \Ls_k} L_k\, \varGamma L_\ell ,
    \end{equation*}
    where $\Ls_k$ collects the edge indexes associated with the two agents linked by edge $k$, excluding the latter, as per the definition in Section~\ref{sec:graphs}.
\end{lemma}
\begin{proof}
See Appendix.
\end{proof}

\begin{proposition}\label{prop:sharing_nodes}
For each $i\in\V$, let 
\begin{subequations}
\label{eq:nodes_selection1}
 \begin{align}
    W_i &:= C_i~, \label{eq:t6P0}\\
    X_i &:= C_i \Xi_1 = \Xi_1 C_i~, \label{eq:t6P1}\\
    Y_i &:=  \sum_{k\in \K_i}  \tfrac{1}{2} \Xi_2 A_k~, \label{eq:t6P2}\\
    Z_i &:= 
      \sum_{k\in \K_i} \tfrac{1}{2}  A_k \Xi_3 A_k + \sum_{k\in \K_i}~\sum_{\ell\in \K_i \setminus \{k\}}A_k \Xi_3 A_{\ell}~, \label{eq:t6P3}
\end{align}   
\end{subequations}
with $A_k$ as per \eqref{eq:whatB}, $C_i$ as per \eqref{eq:C}, and the linear block diagonal $\Xi_1$, $\Xi_2$, $\Xi_3$ as per \eqref{t4Pi}.
If for all $i\in\V$, there exists $\epsilon_i\in\mathbb{R}_{>0}$ such that for all $z\in\L_{2\,}^{2m}$,
  \begin{equation}\label{t6k2} 
        \innerp{\begin{bmatrix}
		I_{2m} \\
		K_i\\
		\end{bmatrix}z}{\left[ \begin{array}{cc} X_i+\epsilon_i W_i & Y_i \\
    	Y^{*}_i & Z_i \end{array} \right]\begin{bmatrix}
		I_{2m} \\
		K_i\\
		\end{bmatrix} z}\leq0,
    \end{equation}
     with $K_i$ as per \eqref{eq:collection_Laplacian_nodes},
 then there exists $\varepsilon\in\mathbb{R}_{>0}$ such that \eqref{eq:stabiliy_G_3b} for all $z\in\L_{2\,}^{2m}$.
\end{proposition}
\begin{proof}
With $Z_i$ as per \eqref{eq:t6P3}, 
\allowdisplaybreaks
\begin{align*}
&\sum_{i\in \V} \innerp{K_i z}{Z_i K_i z}
 \\ 
&=\sum_{i\in \V}  \left \langle  z, K_i \left ( \sum_{k\in \K_i} \tfrac{1}{2} A_k \Xi_3 A_k +\sum_{k\in \K_i}\sum_{\ell \in \K_i \setminus \{k\}}A_k \Xi_3 A_{\ell} \right ) K_i z \right \rangle\\
&=\sum_{i\in \V}  \left \langle  z, \left(
 \sum_{k\in \K_i}  \tfrac{1}{2} L_k  \Xi_3 L_k  + \sum_{k\in \K_i}\sum_{\ell \in \K_i \setminus \{k\}} L_k \Xi_3 L_{\ell} 
\right)z
\right \rangle\\
&= \left \langle  z, \left(
\sum_{i\in \V}  \sum_{k\in \K_i} \tfrac{1}{2} L_k  \Xi_3 L_k  + \sum_{i\in \V}  \sum_{k\in \K_i}\sum_{\ell \in \K_i \setminus \{k\}} L_k \Xi_3 L_{\ell}  \right) z
\right \rangle\\
&= \left \langle  z, \left( \sum_{k\in \M} L_k  \Xi_3 L_k  + \sum_{k\in \M} \sum_{\ell \in \Ls_k} L_k \Xi_3 L_\ell \right) z \right \rangle
=
\innerp{L z}{\Xi_3 L z},
\end{align*}
which implies \eqref{eq:PVRGenNodes_IPZ}. The second equality holds by
the definition of $K_i$ in \eqref{eq:collection_Laplacian_nodes}, 
and the identity
\eqref{eq:BhatMagicC}, whereby $(\sum_{l\in\K_i} L_l) A_k = L_k$ and $ A_\ell (\sum_{l\in\K_i} L_l) = L_\ell $ whenever $k,\ell\in\K_i$.
Both parts of Lemma~\ref{lem:graph_sums_equivalence} are used for the fourth equality, and Lemma~\ref{lem:prop_L} for the final equality.

Similarly, with $X_k$ as per \eqref{eq:t6P1}, given \eqref{eq:propC}, 
$\sum_{i\in \V} X_i = ( \sum_{i\in \V} C_i ) \Xi_1   = \Xi_1$.
As such, $\sum_{i\in \V} \innerp{z}{X_k z} = \innerp{z}{\Xi_1 z}$ for all $z\in\L_{2\,}^{2m}$, which implies \eqref{eq:PVRGenNodes_IPX}. Further, with 
$Y_k$ as per \eqref{eq:t6P2}, in view of  the identity \eqref{eq:BhatMagicD}, linearity of $\Xi_2$, and part~\ref{lem:graph_sums_equivalence1}) of Lemma~\ref{lem:graph_sums_equivalence}, for all $z\in L_{2}^{2m}$,
\begin{align*}
\sum_{i\in \V} Y_i K_i &= \sum_{i\in \V} \tfrac{1}{2} \Xi_2 \,(\sum_{k\in \K_i}A_k) \, (\sum_{\ell\in\K_i} L_\ell ) \nonumber
\\ &= 
 \sum_{i\in \V} \tfrac{1}{2}  \Xi_2 \sum_{k\in\K_i} L_k  
\nonumber
\\
&=
\Xi_2 \sum_{k\in \M} L_k, 
\end{align*}
which implies \eqref{eq:PVRGenNodes_IPY}.
Finally, with
$W_i\succeq 0$ as per \eqref{eq:t6P0}, $\sum_{i\in \V}  W_i = \sum_{i\in \V}  C_i = I_{2m} \succ 0$; see~\eqref{eq:propC}. As such, Lemma~\ref{lem:PVRNodes} applies, and therefore, \eqref{eq:stabiliy_G_3b} for all $z\in\L_{2\,}^{2m}$, as claimed.
\end{proof}
For comparison, the link-wise decomposition from~\cite{MCLinks} is recalled below.
\begin{proposition} (\!\!\cite[Prop.~1]{MCLinks})\label{prop:sharing_couples}
For each $k\in[1:m]$, let 
\begin{align}
    W_k &:= A_k~, \label{eq:t3P0}\\
    X_k &:= \frac{1}{2}\big(A_k\Xi_1 + \Xi_1 A_k\big)~, \label{eq:t3P1}\\
    Y_k &:= \Xi_2A_k~, \label{eq:t3P2}\\
    Z_k &:= 
    A_k\big(\bigoplus_{i=1}^n D_i\big)A_k~, \label{eq:t3P3}
\end{align}
where $A_k$ is defined in \eqref{eq:whatB}, and diagonal $D_i:\L_{2\,}^{m_i}\rightarrow\L_{2\,}^{m_i}$ is such that $\Xi_{3,i}=D_i+S_i$ for some negative semi-definite $S_i:\L_{2\,}^{m_i}\rightarrow\L_{2\,}^{m_i}$, $i\in\V$, with $\Xi_1$, $\Xi_2$, $\Xi_3$ as per \eqref{t4Pi}. 
If for all $k\in\M$, there exists $\epsilon_k\in\mathbb{R}_{>0}$ such that for all $z\in\L_{2\,}^{2m}$
  \begin{equation}\label{PCSk2} 
        \innerp{\begin{bmatrix}
		I_{2m} \\
		L_k\\
		\end{bmatrix}z}{\left[ \begin{array}{cc} X_k+\epsilon_k W_k & Y_k \\
    	Y^{*}_k & Z_k \end{array} \right]\begin{bmatrix}
		I_{2m} \\
		L_k\\
		\end{bmatrix} z}\leq0,
    \end{equation}
 then there exists $\varepsilon\in\mathbb{R}_{>0}$ such that \eqref{eq:stabiliy_G_3b} for all $z\in\L_{2\,}^{2m}$.
\end{proposition}
In Proposition~\ref{prop:sharing_nodes}, $Z_i$ is not restricted to depend on only a suitable diagonal component $D_3$ of $\Xi_3$. This is one aspect of Proposition~\ref{prop:sharing_nodes} that helps mitigate potential conservativeness of the link-based decomposition in Proposition~\ref{prop:sharing_couples}. Further, the determination of $X_i$, $Y_i$, and $Z_i$ in Proposition~\ref{prop:sharing_nodes} depends on model data available to the neighbors of agent $i\in\V$. More specifically, each depends on the dynamics of agent $i$, and its neighbours in $\N_i$, as well as the associated neighbourhood links corresponding to 
the index set $\K_i$.
This stems from the block-diagonal structure of the linear $\Xi_p$, $p \in [1:3]$, in \eqref{t4Pi}, 
and the network structure in each component $L_k$, $k\in \M$, of the sub-system graph Laplacian $L=\sum_{k\in\M}L_k$.
Indeed, the positions and values of the non-zero elements in $K_i Z_i K_i$ reflect the nonzero pattern corresponding to each $L_k$ for $k\in \K_i$, and the same for $Y_i K_i$; see the proof of Lemma~\ref{lem:prop_L} for further detail. The definition of $X_i$, on the other hand, only depends on model data specific to agent $i\in\V$. This additional aspect of Proposition~\ref{prop:sharing_nodes} could potentially reduce the conservativeness of the decentralized conditions, compared to those in Proposition~\ref{prop:sharing_couples}, as each of the latter relies on more localized information that pertains only to each link $k\in\M$.
Since 
the conditions in Proposition~\ref{prop:sharing_nodes} do not require the network to have any particular interconnection structure,
the 
result is generally applicable for the scalable robust stability analysis of sparsely interconnected large-scale systems.

To illustrate the scope for reduced conservativenss provided by Proposition~\ref{prop:sharing_nodes}, consider the following path graph network example with $n=10$ agents; i.e., $m=9$, $m_1=m_{10}=1$ and $m_i=2$ for $i\in[2:9]$. Suppose that the agent dynamics is such that the non-zero entries of $H$ in \eqref{eq:blkdiag} are identical and linear time-invariant, with transfer function $1/(s+25)$. Further, suppose the links $R$ are such that the block elements of the block diagonal $\Delta = (R-I_{18})\circ T$ are sector bounded, with $\Phi_{1,i}=-2m_i\alpha \beta$, $\Phi_{2,i}=\boldsymbol{1}_{m_i}(\alpha+\beta)$ and $\Phi_{3,i}=-2 I_{m_i}$, with $\alpha=-2$, $\beta=0.15$.
Linear Matrix Inequality (LMI) conditions for semi-definite programming-based verification of the monolithic IQC \eqref{eq:stabiliy_G_2}, or \eqref{t6k2} in Proposition,~\ref{prop:sharing_nodes}, or verifying \eqref{PCSk2} in Proposition~\ref{prop:sharing_couples}, can be derived, respectively, via the well-known Kalman-Yakubovich-Popov (KYP) Lemma~\cite{rantzer1996kalman}, given a state-space model for the agent dynamics. The details have been omitted due to space limitations. The LMIs obtained for \eqref{t6k2} in Proposition~\ref{prop:sharing_nodes}, and those for \eqref{eq:stabiliy_G_3b}, are demonstrably feasible, thereby successfully guaranteeing the robust stability of the network. By contrast, the LMIs for \eqref{PCSk2} in Proposition~\ref{prop:sharing_couples} fail to be feasible for a suitable selection of parameters $\epsilon_k>0$, $k\in[1:m]$. 


\section{Conclusions}
\label{sec:conc}
Neighbourhood decentralized robust stability conditions are devised for networked systems in the presence of link uncertainty. This result is based on input-output IQCs that are used to describe the link uncertainties, and ultimately the structured robust stability certificate. An example is used to illustrate the scope for reduced conservativeness compared to existing results. Future work will explore alternative decompositions and comparisons.
It is also of interest to apply the main result to study specific network scenarios where information exchange is impacted by asynchronous time-varying delays, and dynamic quantization (e.g., see~\cite{baillieul2007control}.)

\appendix
\noindent\textbf{\emph{Proof of Lemma~\ref{lem:graph_sums_equivalence}.}} 
Item~\ref{lem:graph_sums_equivalence1}) holds because $\bigcup_{i\in\V} \K_i = \M$, and for $k\in\M$, the cardinality of $\V_k:=\{i\in\V ~|~k\in\K_i\}$ is exactly $2$, so that $ \sum_{i\in\V} \sum_{k\in\K_i} \frac{1}{2} F_k = \sum_{k\in\M} (\tfrac{1}{2}F_k + \tfrac{1}{2}F_k) = \sum_{k\in \M} F_k$.

Further, for all $k\in\M$, the edge index set $$\Ls_k = \bigcup_{i\in\V_k} \K_i\setminus\{k\},$$ and $(\K_i\setminus\{k\})~ \bigcap~(\K_j\setminus\{k\})=\emptyset$ whenever $i, j \in \V_k$ with $i\neq j$. Therefore, 
\begin{align*}
\sum_{i\in\V}~
\sum_{k\in\K_i}~
\sum_{\ell \in\K_i\setminus\{k\}} E_k\circ F_\ell
~
&=
\sum_{k\in\M}~ \sum_{i\in \V_k}~
\sum_{\ell \in\K_i\setminus\{k\}}~ E_k\circ F_\ell\\
&=
\sum_{k\in\M} ~
\sum_{\ell \in \Ls_k} ~ E_k\circ F_\ell,
\end{align*}
which is item~\ref{lem:graph_sums_equivalence4}).     \hfill $\square$

    \noindent\textbf{\emph{Proof of Lemma~\ref{lem:prop_L}.}} 
    First note that $L_k  \varGamma L_\ell = O_{2m}$ for every $k\in\M$, and $\ell \in \M \setminus (\Ls_k\cup\{k\})$.
    Indeed, given the definition of $\Ls_k$, and the block diagonal structure of $\varGamma = \oplus_{i=1}^n \varGamma_i$,
    if the sub-system graph vertex index $r\in[1:2m]$ is such $(\varGamma B_{(\cdot,\ell)})_r \neq 0$, then
    \begin{align*}
        r \notin [\sum_{h\in[1:i-1]}\!m_h+1:\sum_{h\in[1:i]}\!m_h] ~\bigcup~ [\sum_{h\in[1:j-1]}\!m_h+1:\sum_{h\in[1:j]}\!m_h]
    \end{align*}
    with $\{i,j\}=\kappa_{\widetilde{\E}}^{-1}(k)$,
    whereby $B_{(r,k)}=0$. Therefore, $(B_{(\cdot,k)})^\prime (\varGamma B_{(\cdot,\ell)}) = 0$, and thus, $$L_{k}\varGamma L_{\ell} = B_{(\cdot,k)}B_{(\cdot,k)}^\prime \varGamma B_{(\cdot,\ell)}B_{(\cdot,\ell)}^\prime=O_{2m}.$$ 

    Since $\varGamma$ is linear by hypothesis, and pointwise multiplication by each $L_k$ is linear, it follows from the preceding observation that
    \begin{align*}
    L \,\varGamma\, L  &= (\sum_{k\in\M} L_{k}) ~\varGamma ~(\sum_{\ell\in\M} L_{\ell})\\
    &= \sum_{k\in\M} ~ \sum_{\ell\in\M} L_{k}\, \varGamma \, L_{\ell}\\
    &= \sum_{k\in\M}~ \sum_{\ell\in \Ls_k\cup\{k\} } L_{k}\, \varGamma \, L_{\ell} \\
    &= \sum_{k\in\M} (~L_k\, \varGamma\, L_k + \sum_{\ell\in\Ls_k} L_k \, \varGamma\, L_{\ell}~)
    \end{align*}
    as claimed.     \hfill $\square$

\bibliography{Bib}
\bibliographystyle{ieeetr}
\end{document}